\documentclass[10.5pt,a4paper]{article}
\usepackage[latin9]{inputenc}
\usepackage{amsthm}
\usepackage{amsmath}
\usepackage{amssymb}
\usepackage{graphicx}
\usepackage{esint}
\usepackage{amssymb}
\usepackage{float}
\usepackage{caption}
\usepackage{subcaption}
\usepackage{color}
\makeatletter

\theoremstyle{remark}
  \newtheorem{example}{\protect\examplename}
\theoremstyle{plain}
\newtheorem{thm}{\protect\theoremname}
 \theoremstyle{plain}
 \newtheorem{defn}{\protect\definitionname}
\theoremstyle{remark}
\newtheorem{rem}{\protect\remarkname}
\theoremstyle{plain}

\theoremstyle{plain}
\newtheorem{lem}{\protect\lemmaname}
\providecommand{\lemmaname}{Lemma}
\providecommand{\propositionname}{Proposition}
\providecommand{\remarkname}{Remark}
\providecommand{\examplename}{Example}
\providecommand{\definitionname}{Definition}
\providecommand{\theoremname}{Theorem}

\usepackage{amsfonts}\usepackage{color}\usepackage{mathrsfs}\usepackage{wasysym}\usepackage{float}

 \headsep
 \baselineskip
\makeatother

\begin{document}

\title{Stochastic stability and stabilization of a class of state-dependent
jump linear systems\footnote{\textit{preprint submitted to Nonlinear Analysis: Hybrid Systems}}}

\author{Shaikshavali Chitraganti\thanks{CRAN--CNRS UMR 7039, Universit\'{e} de Lorraine, 54500 Vandoeuvre-l\'{e}s-Nancy, France.} \and Samir Aberkane\footnotemark[2] \and Christophe Aubrun\footnotemark[2] }
\date{}

\maketitle
\thispagestyle{empty} 

\begin{abstract}
This paper deals a continuous-time state-dependent jump linear
system, a particular kind of stochastic switching system. In particular,
we consider a situation when the transition rate of the random jump
process depends on the state variable, and addressed the problem of
stochastic stability and stabilization analysis for the proposed system.
Numerically solvable sufficient conditions for the stochastic stability
and stabilization of the proposed system is established in terms of
linear matrix inequalities. The obtained results are illustrated in
numerical examples.
\end{abstract}
\section{Introduction}\label{section:introduction}

Systems subject to random abrupt changes can be modeled by Random Jump linear systems (RJLS) such as manufacturing systems, networked control systems, economics and finance \textit{etc}. RJLS are a special class of hybrid systems, typically, described by a set of classical differential (or difference) equations and a random jump process governing the jumps among them.

When the random jump process of RJLS is assumed to be a finite state time-homogeneous Markovian process with a known transition rate (or probability), then this particular class of systems are widely known as Markov jump linear systems (MJLS) in the literature. The theory of stability, optimal and robust control, as well as important applications of MJLS can be found for instance in \cite{feng}, \cite{mariton}, \cite{kozin}, \cite{yji}, \cite{fangphd} and the references therein. In general, the studies of MJLS assume that the underlying random process is time-homogeneous Markov, which is quite a restrictive assumption. 

In this article, we consider the analysis of RJLS where the random
jump process depends on the state variable. Such class of systems
are referred to as ``state-dependent jump linear systems (SDJLS)''
in this article. The given problem is motivated by following scenarios.
In fault tolerant control systems, the failure rate of a component
generally depends on its age, wear, accumulated stress etc. It is
reasonable to assume that the failure rate of a component at time
$t$ depends on the state of the component at age $t$, see for example
\cite{bergman}. In this case, state variable may be a measure of
wear, accumulated stress of component etc., which affect its failure
rate. As an another scenario, consider a case of stock market with
two regimes: up and down. The transition rate between the regimes
usually result from the state of economy, the general mood of the
investors in the market etc., which can be regarded in general as
the state of the market. Also, for instance, in \cite{motivation1},
the authors dealt with the problem of describing the underlying reasons
for the failure mechanism of a process and modelled the degradation
or wear of the process as a Markov process whose transitions depend
on the state of the process. As an application, the wear of cylinder
lines in a heavy-duty marine diesel engines is considered as a state-dependent
Markov process. Also, let us consider a modelling of macroeconomic
and financial time series. In \cite{motivation2}, a regime-switching
model for the sample path of a time series is examined, where the
transition probabilities between the regimes depend on the state variable.
One can find more examples or scenarios of this kind in the literature.

To the best of the authors' knowledge, only a few works have been
carried out on stability and control of SDJLS. In \cite{yin2009hybrid},
a study of hybrid switching diffusion processes, a kind of state-dependent
jump non-linear systems, has been carried out by treating existence,
uniqueness, stability of the solutions etc. In \cite{sworder1973},
the authors considered that the transition rate of the random jump
process depends on both the state variable and the control input in
such a way that both the state variable and the control input affect
the time scale of the random jump process, thus affecting its transition
rate, and obtained a control policy for a given functional using stochastic
maximum principle. A model for planning and maintenance in flexible
manufacturing system is proposed in \cite{boukas1990}, where the
failure rate of a machine depends on the state variable, and computed
an optimal control using dynamic programming. In \cite{filar2001},
a two-time scale model of production plant is considered as a jump
diffusion model where the failure rate of a machine depends on the
state variable, and obtained an optimal control. 

In this article, we consider the state-dependent transition rates
explicitly as: the transition rates vary depending on which set the
state of the system belongs to. This is a reasonable assumption because
the state of the system at any time belong to one of the predefined
sets, and the transition rate can be considered to have different
values across the predefined sets. The major difference of the current
work in this article with the existing literature on RJLS is that
the random jump process does not follow Markov property. Under the
given assumption that the transition rates vary depending on the set
to which the state of the system belongs to, we prove that the times
at which the change of transition rates occur are stopping times and
accordingly we consider a Dynkin's formula with stopping
Utilizing this formalism, we obtained numerically tractable
sufficient conditions for stochastic stability and stabilization in
terms of linear matrix inequalities (LMIs), though for a restricted
class of SDJLS described in section \ref{section:mathematical}.

The rest of the article is organized as follows: section \ref{section:mathematical}
gives the description of a mathematical model of the SDJLS studied
in this article. In section \ref{section:mainresult}, sufficient
conditions for the stochastic stability and stabilization of the SDJLS
are obtained. In section \ref{section:examples}, numerical examples
are given to illustrate the proposed results, and the concluding remarks
are addressed in section \ref{section:conclusions}.

\textit{Notation} : Let $\mathbb{R}^{n}$ be the n-dimensional real
Euclidean space. $A^{T}$ is the transpose of a matrix $A$. $\lambda_{\text{min}}(A)$
represent the minimum eigenvalue of a matrix $A$. Given two matrices
$L$ and $M$, $L\succ M$ (or $L\prec M$) denotes that the matrix
$L-M$ is positive definite (or negative definite). The standard vector
norm in $\mathbb{R}^{n}$ is indicated by $\Vert.\Vert$ the corresponding
induced norm of a matrix $A$ by $\Vert A\Vert$. Let $\mathbb{Q}$
be the set of rational numbers and $\mathbb{N}_{0}$ be the set of
natural numbers including 0. The operator $\cup$ denotes the union
and $\cap$ denotes intersection. Given two sets $A,\, B,$ $A\setminus B$
denotes the set $A\cap B^{c}.$ The empty set is represented by $\phi$.
For any $a,b\in\mathbb{R},$ $a\wedge b$ represents the minimum of
two numbers $a,b$. $I_{A}(x)$is the standard indicator function
which has a value $1$ if $x\in A$, otherwise has a value $0$. The
mathematical expectation of a random variable $X$ is denoted by $\mathbb{E}[X]$.
Let $g(X_{t})$ be an arbitrary functional of a stochastic process
$X_{t}$; denote $\mathbb{E}[g(X_{t})]_{X_{t}=X}$ as the expectation
of the functional $g(X_{t})$ at $X_{t}=X$.

\section{Mathematical Model}\label{section:mathematical}
Consider a SDJLS in a fixed probability space $\left(\Omega,\mathcal{F},Pr\right)$
\begin{eqnarray}
\dot{x}(t) & = & A_{\theta(t)}x(t),\label{sys:system_old}\\
x(0) & = & x_{0},\notag
\end{eqnarray}
where $x(t)\in\mathbb{R}^{n}$ is the state vector, $x_{0}\in\mathbb{R}^{n}$
is the initial state, $A_{\theta(t)}\in\mathbb{R}^{n\times n}$ be
system matrices which depend on $\theta(t)$. Let $\theta(t)\in S:=\{1,2,\cdots,N\}$,
describing the mode of the system at time $t$, be a finite space
continuous time jump process whose transitions depends on the state
variable $x(t)$ as follows, for $i\ne j$ : 
\begin{align}
Pr\{\theta(t+h)=j/\theta(t)=i,x(t)\} & =\begin{cases}
\lambda_{ij}^{1}h+o(h), & \mathrm{if}\, x(t)\in C_{1},\,\\
\lambda_{ij}^{2}h+o(h), & \mathrm{if}\, x(t)\in C_{2},\,\\
 & \vdots\\
\lambda_{ij}^{K}h+o(h), & \mathrm{if}\, x(t)\in C_{K},\,
\end{cases}\label{eq:rt_intro_old}
\end{align}
where $h>0$ and $\lim\limits _{h\rightarrow 0}\frac{o(h)}{h}=0.$
Let $\mathcal{K}\triangleq\{1,2,\cdots K\}.$ We assume that $C_{1},C_{2},\cdots C_{K}\subseteq\mathbb{R}^{n}$,
$C_{1}\cup C_{2}\cup\cdots C_{K}=\mathbb{R}^{^{n}}$ and $C_{i}\cap C_{j}=\phi$
for any $i\ne j\in\mathcal{K}$. For each $m\in\mathcal{K}$, $\lambda_{ij}^{m}$
is the transition rate of $\theta(t)$ from mode $i$ to mode $j$
with $\lambda_{ij}^{m}\ge 0\ $ for $i\ne j$, and $\lambda_{ii}^{m}=-\sum_{j=1,j\ne i}^{N}\lambda_{ij}^{m}$.
It represents the probability per time unit that $\theta(t)$ makes
a transition from mode $i$ to mode $j$. $o(h)$ is little-$o$ notation
defined by $\lim\limits _{h\rightarrow 0}\frac{o(h)}{h}=0.$ A trivial
remark here is that $m\in\mathcal{K}$ in $\lambda_{ij}^m$ as per (\ref{eq:rt_intro_old}) is a
notation followed in this article, but one should not get confused with the actual power of transition rates.

From the assumption $C_{1}\cup C_{2}\cup\cdots C_{K}=\mathbb{R}^{^{n}}$
and $C_{i}\cap C_{j}=\phi$ for any $i\ne j\in\mathcal{K}$, at any
time $t$, $x(t)$ belongs to one of the sets $C_{l}$, $l\in\mathcal{K}$,
accordingly the transition rate of $\theta(t)$ is $\lambda_{ij}^{l}$from
(\ref{eq:rt_intro_old}). Observe that the transition rate of $\theta(t)$
depends on the state variable $x(t)$, hence we call the system (\ref{sys:system_old})
as SDJLS.

We slightly change the notations of the SDJLS (\ref{sys:system_old})
and the mode $\theta(t)$ (\ref{eq:rt_intro_old}) such that the dealing
of the state dependence becomes simpler. For this purpose consider
$\sigma_t\in\mathcal{K}$, which provide the information of state
variable $x(t)$ at each time $t$ as

\begin{align}
\sigma_t & =\begin{cases}
1, & \mathrm{if}\, x(t)\in C_{1},\,\\
2, & \mathrm{if}\, x(t)\in C_{2},\\
\vdots\\
K, & \mathrm{if}\, x(t)\in C_{K}.
\end{cases}\label{eq:sigma_intro}
\end{align}

Let $r(\sigma_t,t)\in S$ (which is equivalent to $\theta(t)$),
denote the mode of the system at time $t$, be a finite space continuous
time jump process whose transitions depends on $\sigma_t$. Implicitly, $r(\sigma_t,t)$
 depends on the state variable $x(t)$ as follows, for $i\ne j$,
\begin{align}
Pr\{r(\sigma_{t+h},t+h)= & j/r\left(\sigma_t,t \right)=i\}=\begin{cases}
\lambda_{ij}^{1}h+o(h), & \mathrm{if}\,\sigma_t=1,\,\\
\lambda_{ij}^{2}h+o(h), & \mathrm{if}\,\sigma_t=2,\,\\
 & \vdots\\
\lambda_{ij}^{K}h+o(h), & \mathrm{if}\,\sigma_t=K,
\end{cases}\label{eq:rt_intro}
\end{align}
where $\lambda_{ij}^{l}$, for $l\in\mathcal{K}$ is defined in (\ref{eq:rt_intro_old})

Accordingly, we can describe the SDJLS (\ref{sys:system_old}) as
\begin{eqnarray}
\dot{x}(t) & = & A_{r(\sigma_t,t)}x(t),\label{sys:system}\\
x(0) & = & x_{0},\notag
\end{eqnarray}
where $A_{r(\sigma_t,t)}\in\mathbb{R}^{n\times n}$ be system matrices
(which are equivalent to $A_{\theta(t)}$) which depend on $r(\sigma_t,t)$.
From now onwards, we analyse the system (\ref{sys:system}) with jupm
process (\ref{eq:rt_intro}), which is equivalent to analysing the
system (\ref{sys:system_old}) with jump process (\ref{eq:rt_intro_old}).
\begin{rem}
\label{remark:existence} One can observe that the overall system
(\ref{sys:system}) is nonlinear due to the presence of jump process
$r(\sigma_t,t)$. The existence and uniqueness of solution to the
system (\ref{sys:system}) follows directly from theorem 2.1 of \cite{yin2009hybrid}.
\end{rem}
\begin{rem}Observe that, conditioning on $r(\sigma_t,t)=i$, $r \left(\sigma_{t+h},t+h\right)$
depends on $x(t)$ for any $h>0$, and from (\ref{sys:system}), which
in turn depends on $r(\sigma_{s},s)$, $s<t$. Hence $r(\sigma_t,t)$
is not a Markov process. However $\left(x(t),r(\sigma_t,t),\sigma_t\right)$
is a joint Markov process. This point is stated and proved in the
following lemma. \end{rem}

\begin{lem}
\label{lemma:joint_Markov_process} $\left(x(t),r(\sigma_t,t),\sigma_t\right)$
is a joint Markov process.\end{lem}
\begin{proof}
Given in the Appendix.
\end{proof}
The solution to (\ref{sys:system}) can be constructed as presented
in the sequel. In that direction, we define first exit times from
the sets $C_{j},$ for $j\in\mathcal{K}$. We use a convention $\text{inf}\ \phi\triangleq\infty$. 
\begin{itemize}
\item Step $0$: Let $x(0)\in C_{i_{0}}$, where $i_{0}\in\mathcal{K}$.
Define $\tau_{0}$ as the first exit time from $C_{i_{0}}$ as 
\[
\tau_{0}=\mathrm{inf}\{t\ge 0:\,\Phi_{i_{0}}(t,0)x(0)\notin C_{i_{0}}\}.
\]

\item Step 1: Let $x(\tau_{0})\in C_{i_{1}}$, where $i_{1}\ne i_{0}$,
$i_{1}\in\mathcal{K}$. Define $\tau_{1}$ as the first exit time
from $C_{i_{1}}$ after $\tau_{0}$ as 
\[
\tau_{1}=\mathrm{inf}\{t\ge\tau_{0}:\,\Phi_{i_{1}}(t,\tau_{0})\Phi_{i_{0}}(\tau_{0},0)x(0)\notin C_{i_{1}}\}.
\]

\item Step 2: Let $x(\tau_{1})\in C_{i_{2}}$, where $i_{2}\ne i_{1}$,
$i_{2}\in\mathcal{K}$. Define $\tau_{2}$ as the first exit time
from $C_{i_{2}}$ after $\tau_{1}$ as 
\begin{align*}
\tau_{2}=\mathrm{inf} & \{t\ge\tau_{1}:\,\Phi_{i_{2}}(t,\tau_{1})\Phi_{i_{1}}(\tau_{1},\tau_{0})\Phi_{i_{0}}(\tau_{0},0)x(0)\notin C_{i_{2}}\}.
\end{align*}

\end{itemize}
In general, at any step $m$, given $\tau_{m-1}$, $i_{m-1}\in\mathcal{K}$
of the previous step $m-1$ which is defined in a similar manner above,
\begin{itemize}
\item Step $m$: Let $x(\tau_{m-1})\in C_{i_{m}}$, where $i_{m}\ne i_{m-1}$,
$i_{m}\in\mathcal{K}$. Define $\tau_{m}$ as the first exit time
from $C_{i_{m}}$ after $\tau_{m-1}$ as 
\begin{align}
\tau_{m}= & \mathrm{inf}\{t\ge\tau_{m-1}:\,\Phi_{i_{m}}(t,\tau_{m-1})\Phi_{i_{m-1}}(\tau_{m-1},\tau_{m-2})\cdots\Phi_{i_{0}}(\tau_{0},0)x(0)\notin C_{i_{m}}\},\label{eq:tau}
\end{align}

\end{itemize}
where the random flows $\Phi_{.}(.,.)$ are defined in the sequel.
We describe, in general, one of the random flows $\Phi_{i_{m}}(t,\tau_{m-1})$
in (\ref{eq:tau}), at step $m$, with $x(\tau_{m-1})\in C_{i_{m}}$,
using which any random flow of (\ref{sys:system}) can be described
in a similar fashion. At step $m$, during the interval $[\tau_{m-1},\tau_{m})$,
with $\tau_{-1}\triangleq 0,$: let $n_{m}\in\mathbb{N}_{0}$ be the
number of regime transitions of $r(\sigma_t,t)$ ; let $\{r_{0}^{m},r_{1}^{m},\cdots r_{n_{m}}^{m}\}\in S$
be the sequence of regimes visited by $r(\sigma_t,t)$; let $\left\{ T_{0}^{m},T_{1}^{m},\cdots T_{n_{m}}^{m}\right\} \in[\tau_{m-1},\tau_{m})$
be the successive sojourn times of $r(\sigma_t,t)$, which are independent
exponentially distributed random variables with parameter $\lambda_{ij}^{i_{m}}$.
Let $S_{n_{m}}^{m}\triangleq\sum_{l=0}^{n_{m}-1}T_{l}^{m}.$ Then
$\Phi_{i_{m}}(t,\tau_{m-1})$ is given by,

\begin{align*}
\Phi_{i_{m}}(t,\tau_{m-1}) & =\begin{cases}
\text{e}^{A_{r_{n_{m}}^{m}}\left(t-S_{n_{m}}^{m}-\tau_{m-1}\right)}\text{e}^{A_{r_{n_{m}-1}^{m}}T_{n_{m}-1}^{m}}\cdots\text{e}^{A_{r_{1}^{m}}T_{1}^{m}}\text{e}^{A_{r_{0}^{m}}T_{0}^{m}}, & \text{if }n_{m}\ge 1,\\
\text{e}^{A_{r_{n_{m}}^{m}}\left(t-\tau_{m-1}\right)}, & \text{if }n_{m}=0.
\end{cases}
\end{align*}

\begin{rem}
\label{rem:sigma_r_reform}Notice that, from step $0$, step $1$,
$\cdots$ step $m$,$\cdots,$ $\sigma_t$ and $r(\sigma_t,t)$
can be described alternatively for $t\ge 0$ by 
\begin{align*}
\sigma_t & =\begin{cases}
i_{0}, & \mathrm{if}\, t\in[0,\tau_{0}),\,\\
i_{1}, & \mathrm{if}\, t\in[\tau_{0},\tau_{1}),\,\\
 & \vdots\\
i_{m}, & \mathrm{if}\, t\in[\tau_{m-1},\tau_{m}),\,\\
 & \vdots
\end{cases}
\end{align*}
and
\begin{align}
Pr\{r(\sigma_{t+h},t+h)=j/r(\sigma_t,t)=i\}=\begin{cases}
\lambda_{ij}^{i_{0}}h+o(h), & \mathrm{if}\, t\in[0,\tau_{0}),\,\\
\lambda_{ij}^{i_{1}}h+o(h), & \mathrm{if}\, t\in[\tau_{0},\tau_{1}),\,\\
 & \vdots\\
\lambda_{ij}^{i_{m}}h+o(h), & \mathrm{if}\, t\in[\tau_{m-1},\tau_{m}),\,\\
 & \vdots
\end{cases}\notag
\end{align}
where $\{i_{0},i_{1},\cdots,i_{m},\cdots\}\in\mathcal{K}$ and $i\ne j$,
$i,j\in S.$
\end{rem}

\begin{rem}
Though the alternative reformulations of $\sigma_t$ and $r(\sigma_t,t)$
in remark \ref{rem:sigma_r_reform} seems not much useful at this
point, but the results of section \ref{section:mainresult} will be
based on these reformulations.
\end{rem}
\begin{rem}\label{remark:gen_mjls}From step 0, given $x(0)\in C_{i_{0}},$
for $i_{0}\in\mathcal{K},$ if $\tau_{0}=\infty,$ then $x(t)\in C_{i_{0}}$
for all $t\ge 0.$ In this case the overall system (\ref{sys:system})
is equivalent to time-homogeneous MJLS with jump process being time-homogeneous
Markov with parameter $\lambda_{ij}^{i_{0}}.$\end{rem}

\noindent We define the stopping time in the sequel and prove that
the first exit times $\tau_{0}$, $\,\tau_{1},\,$ $\tau_{2},\cdots$
are the stopping times.
\begin{defn}
\noindent Let $\left(\Omega,\mathcal{F},\mathcal{G}_{t},Pr\right)$ be a filtered
probability space, then a random variable $\tau:\,\Omega\rightarrow[0,\infty]$
(it may take the value $\infty$) is called a stopping time
if $\{\tau\le t\}\in\mathcal{G}_{t}$ for any $t\ge 0$, i.e; the event
$\{\tau\le t\}$ is $\mathcal{G}_{t}$-measurable, which implies the
event $\{\tau\le t\}$ is completely determined by the knowledge of
$\mathcal{G}_{t}$. \end{defn}

The following lemma shows that the first exit times given above are in fact stopping times.
\begin{lem}
\label{lemma:stopping_time} The first exit times $\tau_{0}$,$\,\tau_{1},\,$$\tau_{2},\cdots$ described in step 0, step 1, step 2, $\cdots$ are stopping times.\end{lem}
\begin{proof}
Given in the Appendix.
\end{proof}

Based on lemma \ref{lemma:joint_Markov_process} and lemma \ref{lemma:stopping_time}, we provide a Dynkin's formula that will be used in the next section.

\begin{defn} Let $\left(x(t),r(\sigma_t,t),\sigma_t\right)$ be a Markov process
and $\tau_{0}$,$\tau_{1}$,$\tau_{2},....$ are stopping times. Let
$\xi(t)\triangleq \left(x(t),r(\sigma_t,t),\sigma_t\right)$. For any suitable
Lyapunov function $V(\xi(t))$, the Dynkin's formula can described
as \cite{sri}, \cite{karimi},

\begin{align}
\mathbb{E}[V(\xi(t))|\xi(0)]-V(\xi(0)) & =\mathbb{E}\left[\sum_{j=0}^{j^{*}}\int_{t\wedge\tau_{j-1}}^{t\wedge\tau_{j}}\mathcal{L}V(\xi(s))ds|\xi(s)\right]\notag\\
 & =\sum_{j=0}^{j^{*}}\mathbb{E}\left[\int_{t\wedge\tau_{j-1}}^{t\wedge\tau_{j}}\mathcal{L}V(\xi(s))ds|\xi(s)\right],\label{eq:dynkin}
\end{align}
where $\mathcal{L}V(\xi(t))$ is the infinitesimal generator of $V(\xi(t))$.
Here $\tau_{-1}=0$, and $j=0,1,\cdots j^{*}$, where $j^{*}\in[0,\infty]$
and $\tau_{j^{*}}\le\infty$. \end{defn}
In general,   $\mathcal{L}V(\xi(t))$ can be understood as the average time rate of change of the function $V(\xi(t))$ given $\xi(t)$ at time $t$. Also observe that, since $\xi(t)$ is Markov process, for any $t\ge 0$, the expectation terms in (\ref{eq:dynkin}) are
conditioned on $\xi(t)$, instead of the natural filtration of $\xi(t)$ on the interval $[0,t]$.

\section{Main Results\label{section:mainresult} }
In this section, we present sufficient conditions for stochastic stability
and stabilization of the system (\ref{sys:system}). 

\subsection{Stochastic stability}
We begin with a definition of stochastic stability,

\begin{defn} For system (\ref{sys:system}), the equilibrium point
$0$ is \textit{stochastically stable} if, for any $x_{0}\in\mathbb{R}^{n}$
and any $r(\sigma_{0},0)\in S:=\{1,2,\cdots,N\}$ and $\sigma_{0}\in\mathcal{K}:=\{1,2,\cdots K\}$,
\begin{align*}
\mathbb{E}\left[\int_{0}^{\infty}||x(t)||^{2}dt\right]<\infty.
\end{align*}
\end{defn}

\noindent We now provide a sufficient condition for stochastic stability. 
\begin{thm}
\noindent \label{theorem:1} The system (\ref{sys:system}) is stochastically
stable if there exist positive definite matrices $P_{i}\succ 0$, $W_{\kappa i}\succ 0$ for
all $i\in S$ and for all $\kappa\in\mathcal{K}$, satisfying 

\begin{equation}
A_{i}^{T}P_{i}+P_{i}A_{i}+\sum_{j=1}^{N}\lambda_{ij}^{\kappa}P_{j}=-W_{\kappa i}, \label{lmi:1}
\end{equation}
 where $\lambda_{ij}^{\kappa}$ is defined in (\ref{eq:rt_intro_old}).
\end{thm}
\noindent \textit{Proof}: Consider a $V\left(x(t),r(\sigma_t,t),\sigma_t\right)=x^{T}(t)P_{r(\sigma_t,t)}x(t)$,
which is quadratic and positive in $x(t)$, hence a Lyapunov candiate
function. Let the infinitesimal generator of $V\left(x(t),r(\sigma_t,t),\sigma_t\right)$, for any $i\in S$ and for any $\kappa \in \mathcal{K}$, be given by,
\begin{align*}
\mathcal{L}V\big( & x(t),r(\sigma_t=\kappa,t)=i,\sigma_t=\kappa\big)\\
 & =\lim_{h\rightarrow 0}\frac{1}{h}\bigg\{\mathbb{E}\Big[V\big(x(t+h),r(\sigma_{t+h},t+h),\sigma_{t+h}\big)|\big(x(t),r(\sigma_t=\kappa,t)=i,\sigma_t=\kappa\big)\Big]\\
 & \;\qquad\;-V\big(x(t),r(\sigma_t=\kappa,t)=i,\sigma_t=\kappa\big)\bigg\}\\
 & =x^{T}(t)\Big[A_{i}^{T}P_{i}+P_{i}A_{i}+\sum_{j=1}^{N}\lambda_{ij}^{\kappa}P_{j}\Big]x(t).
\end{align*}
The above derivation is quite straight forward and follows the similar
approach as given in \cite{boukasbook}, \cite{costabook2012} for
example. Hence, by (\ref{lmi:1}),
\begin{align}
\hspace{-2.3cm}\mathcal{L}V\big(x(t),r(\sigma_t=\kappa,t)=i,\sigma_t=\kappa\big) & =-x^{T}(t)W_{\kappa i}x(t) \notag\\
 & \le-\min_{\kappa\in\mathcal{K},i\in S}\{\lambda_{\text{min}}(W_{\kappa i})\}x^{T}(t)x(t).\label{eq:lvmin}
\end{align}
From (\ref{eq:dynkin}), consider for any $i_{0}\in\mathcal{K},$
\begin{align*}
\mathbb{E} & \Big[V\big(x(t),r(\sigma_t,t),\sigma_t\big)|\big(x(0),r(\sigma_{0}=i_{0},0),\sigma_{0}=i_{0}\big)\Big]-V\big(x(0),r(\sigma_{0}=i_{0},0),\sigma_{0}=i_{0}\big)\\
 & =\sum_{j=0}^{j^{*}}\mathbb{E}\left[\int_{t\wedge\tau_{j-1}}^{t\wedge\tau_{j}}\mathcal{L}V\big(x(s),r(\sigma_{s},s),\sigma_{s}\big)ds|(x(s),r(\sigma_{s},s),\sigma_{s})\right],
\end{align*}
 where $j^{*},$ $\tau_{-1}$ and $\tau_{j^{*}}$ are given
in (\ref{eq:dynkin}). Let $\{i_{0},i_{1},i_{2},\cdots\}\in\mathcal{K}$
be the successive states visited by $\sigma_t$ similar to remark
\ref{rem:sigma_r_reform}. Then 
\begin{align*}
\mathbb{E}  \Big[V\big(x(t),& r(\sigma_t,t),\sigma_t\big)|\big(x(0),r(\sigma_{0}=i_{0},0),\sigma_{0}=i_{0}\big)\Big]-V\big(x(0),r(\sigma_{0}=i_{0},0),\sigma_{0}=i_{0}\big)\\
= & \mathbb{E}\left[\int_{0}^{\tau_{0}}\mathcal{L}V\big(x(s),r(\sigma_{s}=i_{0},s),\sigma_{s}=i_{0}\big)ds|\big(x(s),r(\sigma_{s}=i_{0},s),\sigma_{s}=i_{0}\big)\right]\\
 & +\mathbb{E}\left[\int_{\tau_{0}}^{\tau_{1}}\mathcal{L}V\big(x(s),r(\sigma_{s}=i_{1},s),\sigma_{s}=i_{1}\big)ds|\big(x(s),r(\sigma_{s}=i_{1},s),\sigma_{s}=i_{1}\big)\right]\\
 & +\cdots+\mathbb{E}\left[\int_{t\wedge\tau_{j^{*}-1}}^{t\wedge\tau_{j^{*}}}\mathcal{L}V\big(x(s),r(\sigma_{s},s),\sigma_{s}\big)ds|(x(s),r(\sigma_{s},s),\sigma_{s})\right].
\end{align*}
By (\ref{eq:lvmin}), 
\begin{align*}
\mathbb{E} & \Big[V\big(x(t),r(\sigma_t,t),\sigma_t\big)|\big(x(0),r(\sigma_{0}=i_{0},0),\sigma_{0}=i_{0}\big)\Big]-V\big(x(0),r(\sigma_{0}=i_{0},0),\sigma_{0}=i_{0}\big)\\
 & \hspace{1.5cm}\le-\min_{\kappa\in\mathcal{K},i\in S}\{\lambda_{\text{min}}(W_{\kappa i})\}\bigg(\mathbb{E}\left[\int_{0}^{\tau_{0}}x^{T}(s)x(s)ds\right]+\mathbb{E}\left[\int_{\tau_{0}}^{\tau_{1}}x^{T}(s)x(s)ds\right]+\\
 &\hspace{2cm} \cdots+\mathbb{E}\left[\int_{t\wedge\tau_{j^{*}-1}}^{t\wedge\tau_{j^{*}}}x^{T}(s)x(s)ds\right]\bigg).
\end{align*}
By denoting $\sum_{j=0}^{j^{*}}\int_{t\wedge\tau_{j-1}}^{t\wedge\tau_{j}}=\int_{0}^{t}$,
one obtains, 
\begin{align*}
\mathbb{E}\Big[V\big(x(t),r(\sigma_t, & t),\sigma_t\big)|\big(x(0),r(\sigma_{0}=i_{0},0),\sigma_{0}=i_{0}\big)\Big]-V\big(x(0),r(\sigma_{0}=i_{0},0),\sigma_{0}=i_{0}\big)\\
\le & -\min_{\kappa\in\mathcal{K},i\in S}\{\lambda_{\text{min}}(W_{\kappa i})\}\mathbb{E}\left[\int_{0}^{t}x^{T}(s)x(s)ds\right].
\end{align*}
By rearranging the terms, 
\begin{align*}
\min_{\kappa\in\mathcal{K},i\in S}&\{\lambda_{\text{min}}(W_{\kappa i})\}\mathbb{E}\left[\int_{0}^{t}x^{T}(s)x(s)ds\right] \\
  &\le V\big(x(0),r(\sigma_{0}=i_{0},0),\sigma_{0}=i_{0}\big)\\
  &\hspace{2cm} -\mathbb{E}\Big[V\big(x(t),r(\sigma_t,t),\sigma_t\big)| \big(x(0),r(\sigma_{0}=i_{0},0),\sigma_{0}=i_{0}\big)\Big]\\
 & \le V\big(x(0),r(\sigma_{0}=i_{0},0),\sigma_{0}=i_{0}\big).
\end{align*}

\noindent Thus, 
\begin{align*}
\mathbb{E}\left[\int_{0}^{t}x^{T}(s)x(s)ds\right]\le\frac{V\big(x(0),r(\sigma_{0}=i_{0},0),\sigma_{0}=i_{0}\big)}{\min_{\kappa\in\mathcal{K},i\in S}\{\lambda_{\text{min}}(W_{\kappa i})\}}.
\end{align*}
By letting $t\rightarrow\infty$, 
\begin{align*}
\mathbb{E}\left[\int_{0}^{\infty}x^{T}(s)x(s)ds\right]\le\frac{V\big(x(0),r(\sigma_{0}=i_{0},0),\sigma_{0}=i_{0}\big)}{\min_{\kappa\in\mathcal{K},i\in S}\{\lambda_{\text{min}}(W_{\kappa i})\}}<\infty.
\end{align*}
Thus the system (\ref{sys:system}) is stochastically stable.\hfill{}
$\square$

\begin{rem}Similar to remark \ref{remark:gen_mjls}, from step 0,
given $x(0)\in C_{i_{0}},$ for $i_{0}\in\mathcal{K},$ if $\tau_{0}=\infty,$
then the overall system (\ref{sys:system}) is equivalent to time-homogenous
MJLS. And LMIs (\ref{lmi:1}) are equivalent to the LMIs given in
theorem 1 of \cite{boukasbook} with the transition rate $\lambda_{ij}^{i_{0}}.$\end{rem}

\subsection{Stochastic stabilization}

In this section we provide a sufficient condition for stochastic stabilization
with state feedback controller. Consider the system (\ref{sys:system})
with control input $u(t)\in\mathbb{R}^{m}$ 
\begin{equation}
\dot{x}(t)=A_{r(\sigma_t,t)}x(t)+B_{r(\sigma_t,t)}u(t),\label{sys:system_withipold}
\end{equation}
where $B_{r(\sigma_t,t)}\in\mathbb{R}^{n\times m}$.

We make use of the theorem \ref{theorem:1} to design a state-feedback
stabilizing controller such that the system (\ref{sys:system_withipold})
is stochastically stable. We assume that the system mode $r(\sigma_t,t)$
is available in real time, which lead to a state feedback control
law of the form 
\begin{equation}
u(t)=K_{r(\sigma_t,t)}x(t),\quad K_{r(\sigma_t,t)}\in\mathbb{R}^{m\times n}\label{eq:stateFB}
\end{equation}
then the overall system is given by

\begin{equation}
\dot{x}(t)=\tilde{A}_{r(\sigma_t,t)}x(t),\label{sys:system_withip}
\end{equation}
where $\tilde{A}_{r(\sigma_t,t)}=A_{r(\sigma_t,t)}+B_{r(\sigma_t,t)}K_{r(\sigma_t,t)}.$

\begin{rem}Since the system (\ref{sys:system_withip}) is identical
to the system (\ref{sys:system}), the existence and uniqueness of the solutions
of (\ref{sys:system_withip}) follow from remark \ref{remark:existence}.\end{rem}

The following theorem provides a sufficient condition for the existence
of a stabilizing controller of the form (\ref{eq:stateFB}).

\begin{thm}\label{theorem:stabilization_SS} Consider the system
(\ref{sys:system_withipold}) with $\sigma_t$ and $r(\sigma_t,t)$
described in (\ref{eq:sigma_intro}) and (\ref{eq:rt_intro}). If
there exist matrices $X_{i}\succ0$, and $Y_{i}$, for each $i\in S$
such that, 
\begin{equation}
\begin{bmatrix}J_{i} & \mathcal{M}_{\kappa i}\\
\ast & -\mathcal{X}_{i}
\end{bmatrix}\prec0,\label{lmi:stabilization}
\end{equation}
for each $\kappa\in\mathcal{K}$,  where 
\begin{align*}
J_{i}= & X_{i}A_{i}^{T}+Y_{i}^{T}B_{i}^{T}+A_{i}X_{i}+B_{i}Y_{i}+\left(\sum_{j=1}^{K}\lambda_{ii}^{j}I_{\{\kappa=j\}}\right)X_{i},\\
\mathcal{M}_{\kappa i}= & [\sqrt{\lambda_{i1}^{\kappa}}X_{i}\,\cdots\sqrt{\lambda_{ii-1}^{\kappa}}X_{i}\,\sqrt{\lambda_{ii+1}^{\kappa}}X_{i}\,\cdots\sqrt{\lambda_{iN}^{\kappa}}X_{i}],\\
\mathcal{X}_{i}= & \mathrm{diag}\{X_{1}\cdots X_{i-1},X_{i+1}\cdots X_{N}\},
\end{align*}
then the system (\ref{sys:system_withipold}) is stochastically stabilized
by (\ref{eq:stateFB}), and the stabilizing controller is given by
\begin{equation}
K_{i}=Y_{i}X_{i}^{-1}.\label{controller:stateFB}
\end{equation}
\end{thm} 
\begin{proof}
The proof is an immediate extension of theorem 11 of \cite{boukasbook}.
\end{proof}

\section{Illustrative Examples}\label{section:examples} 
In this section, two numerical examples
are presented to illustrate the proposed results.\\
\begin{example}\label{example:1}

\noindent Consider the SDJLS (\ref{sys:system_old})
with $x(t)\triangleq[x_{1}(t),x_{2}(t)]^{T}\in\mathbb{R}^{2}$. Let
$\theta(t)\in S:=\{1,2\}$, be the state-dependent jump process given
by (\ref{eq:rt_intro_old}), with $C_{1}\triangleq\{x(t)\in\mathbb{R}^{2}:x_{1}^{2}(t)+x_{2}^{2}(t)<3\}$,
$C_{2}\triangleq\{x(t)\in\mathbb{R}^{2}:x_{1}^{2}(t)+x_{2}^{2}(t)\ge3\}$,
and for $\mathcal{K}:=\{1,2\}$ the transition rate matrices of $\theta(t)$
are given by

\begin{align*}
\left(\lambda_{ij}^{1}\right)_{2\times2}= \begin{bmatrix}-2 & 2\\
2 & -2
\end{bmatrix}, & \left(\lambda_{ij}^{2}\right)_{2\times2}= \begin{bmatrix}-4 & 4\\
4 & -4
\end{bmatrix}.
\end{align*}
Let 
\begin{align*}
A_{1}=\begin{bmatrix}-1 & 5\\
-0.5 & 0.9
\end{bmatrix},\, & A_{2}=\begin{bmatrix}-4 & 2\\
-2 & 0.1
\end{bmatrix}.
\end{align*}
Then, from theorem \ref{theorem:1}, the LMIs (\ref{lmi:1}), are
satisfied with
\begin{align*}
P_{1}=\begin{bmatrix}0.3787 & -0.4069\\
-0.4069 & 2.2977
\end{bmatrix},\, & P_{2}=\begin{bmatrix}0.3891 & -0.6203\\
-0.6203 & 1.9226
\end{bmatrix}.
\end{align*} 
Hence, the system (\ref{sys:system_old}) is stochastically stable.
With $\theta(0)=1$, $x(0)=[-1,1]^{T}$, a sample $\theta(t)$ with
the corresponding stopping times $\tau_{0}$, $\tau_{1}$, $\cdots$
are given in step 0, step 1, $\cdots$are plotted in figure \ref{fig:rt};
the corresponding sample state trajectories of the system are shown
in figure \ref{fig:states}.

%

\begin{figure}[h]
        \centering
        \begin{subfigure}[b]{0.45\textwidth}
                \includegraphics[width=\textwidth,height=5cm]{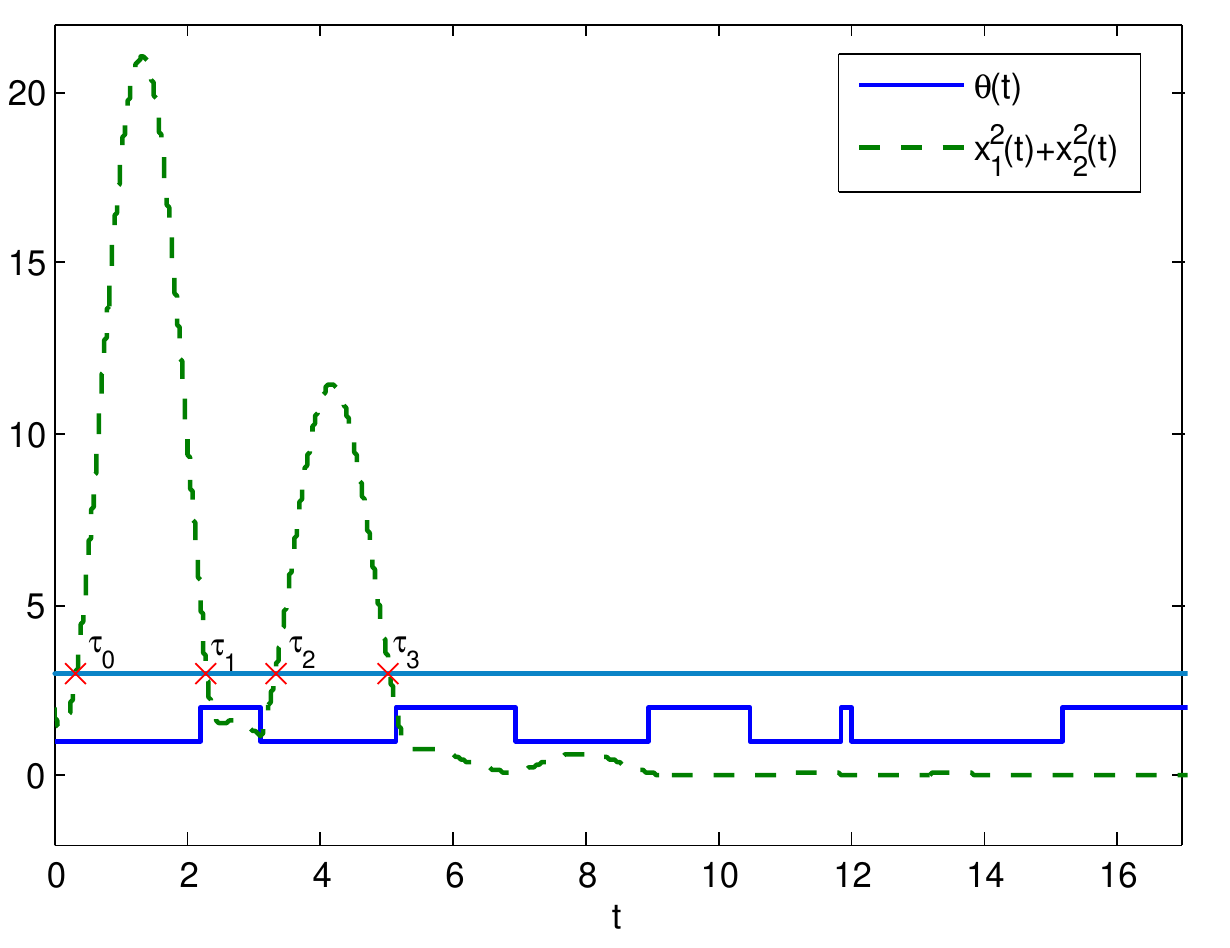}
				\caption{A sample $\theta(t)$ and the corresponding stopping times }
               \label{fig:rt} 
        \end{subfigure}%
        ~ 
        \begin{subfigure}[b]{0.45\textwidth}
                \includegraphics[width=\textwidth,height=5cm]{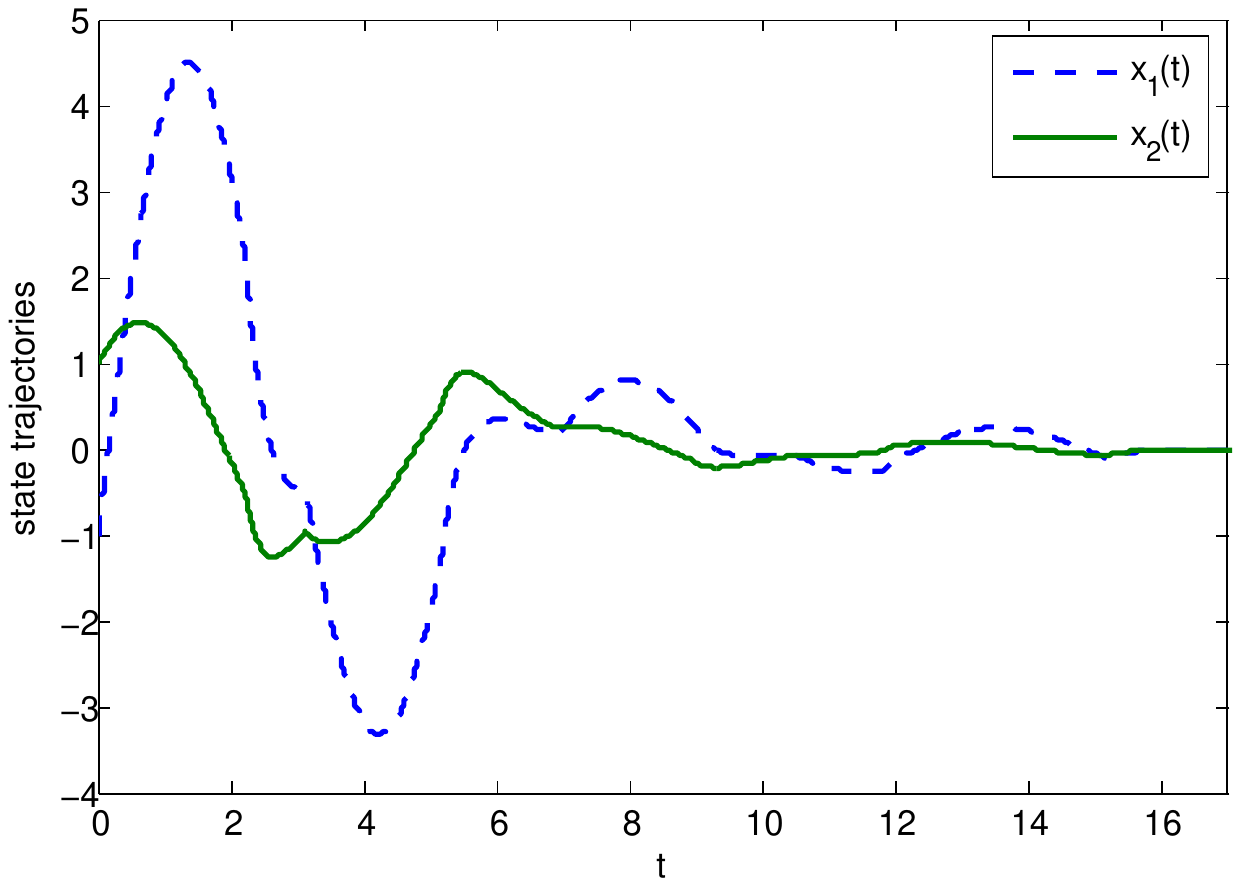}
				\caption{State trajectories of the system: a single sample path }
                \label{fig:states}
        \end{subfigure}
             \caption{Plots of example 1 }\label{fig:example1}
\end{figure}
\end{example}

\begin{example} \label{example:2}
 In this example we utilize the result of theorem
\ref{theorem:stabilization_SS} to synthesize a state feedback controller
of form (\ref{eq:stateFB}) that stochastically stabilizes (\ref{sys:system_withipold}).
Let $\theta(t)\in S:=\{1,2\}$ consists of the same parameters as
example \ref{example:1}. Let
\begin{align*}
A_{1}=\begin{bmatrix}-1 & 2\\
-2 & 1
\end{bmatrix}, A_{2}=\begin{bmatrix}1 & 2\\
2 & 1
\end{bmatrix}, B_{1}=\begin{bmatrix}1 \\
3
\end{bmatrix}, B_{2}=\begin{bmatrix}-5 \\
6
\end{bmatrix}.
\end{align*} 
By solving the LMIs (\ref{lmi:stabilization}),
we obtain:
\begin{align*}
X_{1} & =\begin{bmatrix}0.343 & -0.365\\
-0.365 & 0.3973
\end{bmatrix}, X_{2}  =\begin{bmatrix}0.3998 & -0.4203\\
-0.4203 & 0.4462
\end{bmatrix},\\
Y_{1} & =\begin{bmatrix}0.2597 & -0.5748\end{bmatrix},\hspace{0.11cm} Y_{2}  =\begin{bmatrix}-0.0385 & 0.0032\end{bmatrix},
\end{align*}
and the feedback gains are obtained as: 
\begin{align*}
K_{1}=\begin{bmatrix}-35.4961-34.0615\end{bmatrix},\quad K_{2}=\begin{bmatrix}-8.9022 & -8.3779\end{bmatrix}.
\end{align*} 
 With $\theta(0)=1$, $x(0)=[-1,1]^{T}$, a single sample path simulation corresponding
to a realization of $\theta(t)$ is given in figure \ref{fig:rt_fb};
a sample state trajectories of the closed loop system resulting from
the obtained controller are shown in figure \ref{fig:states_fb}.
It can be observed that the closed loop system is stochastically stable.

%
%

%
%

\begin{figure}[h]
\centering
                \begin{subfigure}[b]{0.45\columnwidth}
                \includegraphics[width=\columnwidth,height=5cm]{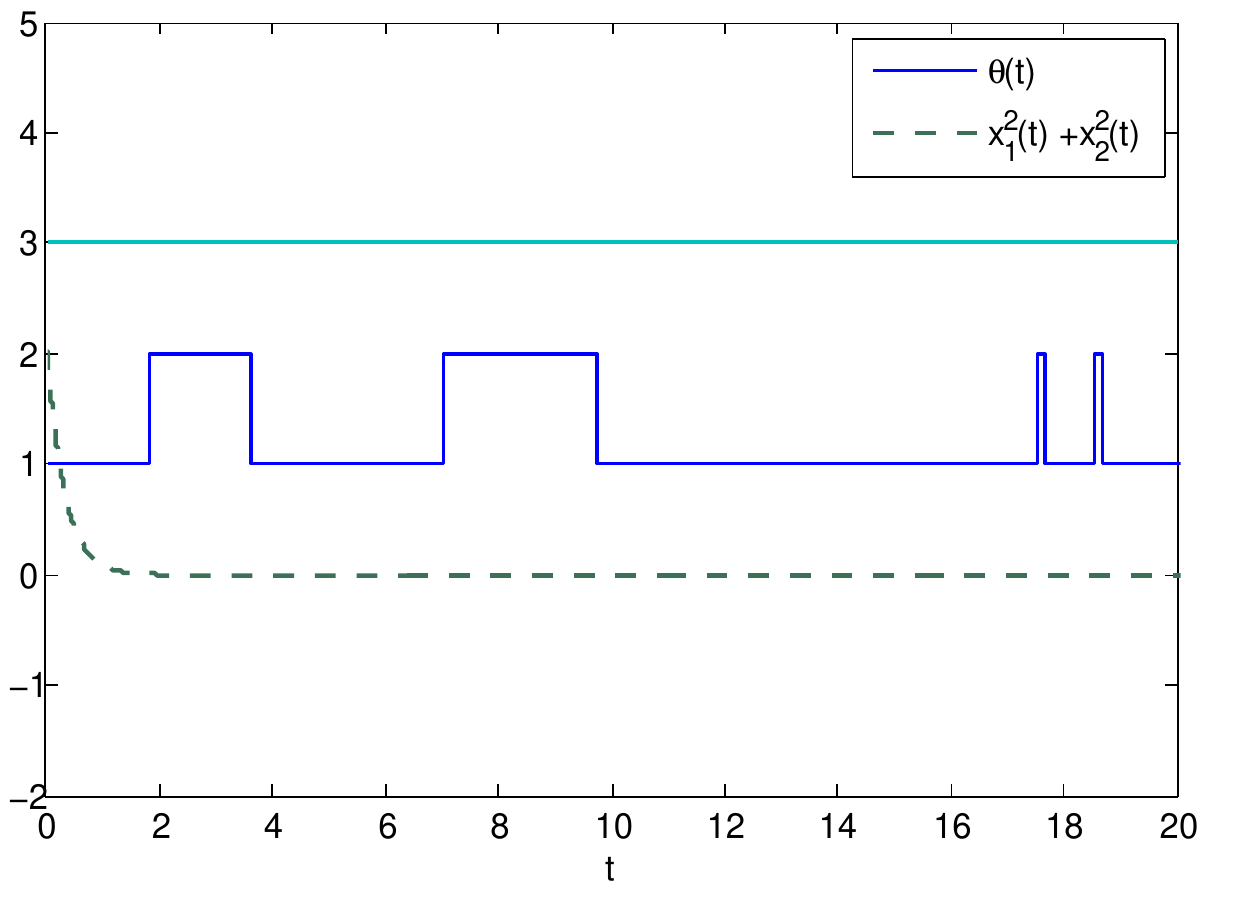}
	\caption{A sample $\theta(t)$ \textcolor{white}{i ttttt ttt tt err rrr rrr rr tt tt tt tt r}	
			}
               \label{fig:rt_fb} 
        \end{subfigure}%
        ~ 
       \begin{subfigure}[b]{0.45\columnwidth}
                \includegraphics[width=\columnwidth,height=5cm]{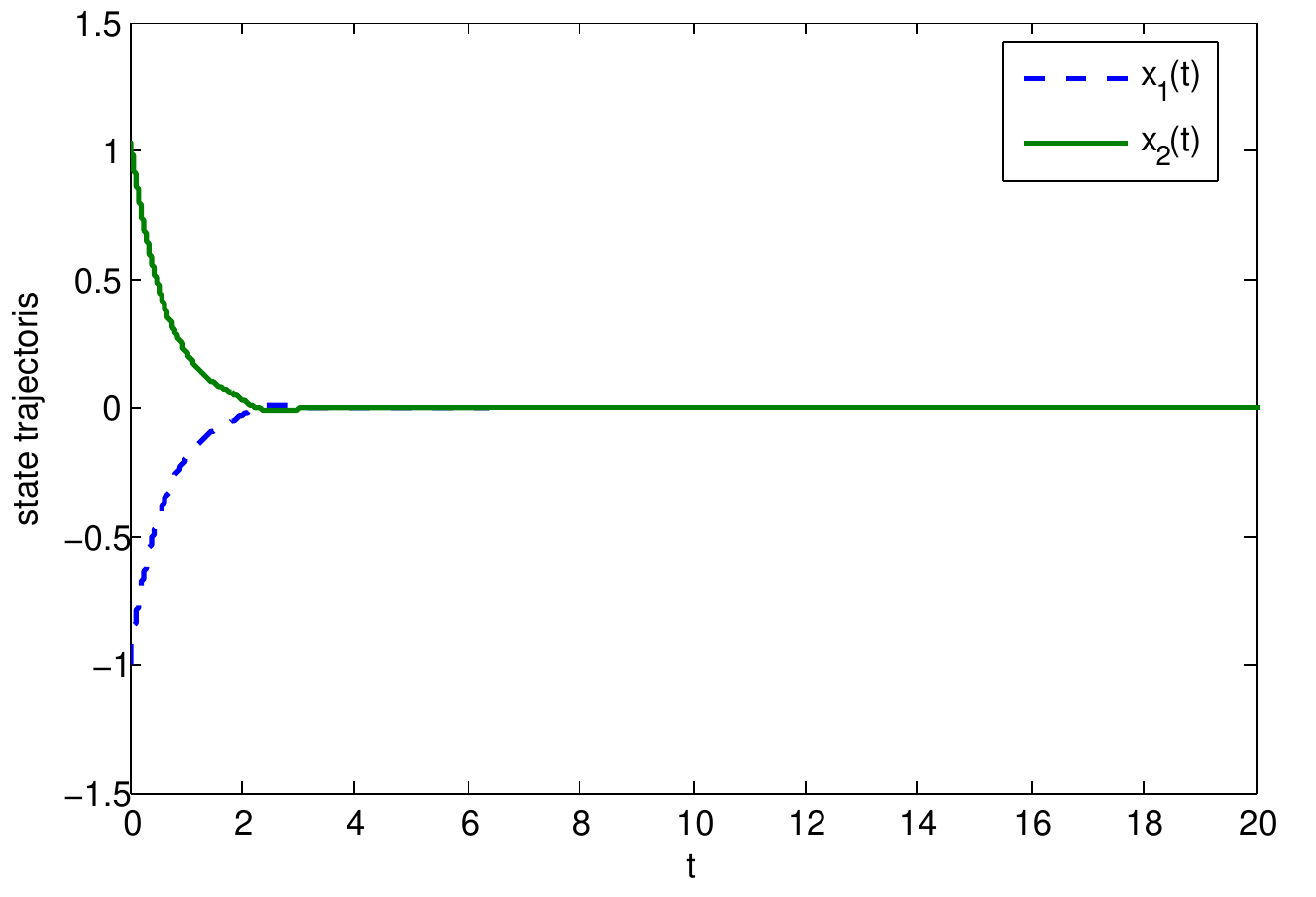}
	\caption{State trajectories of the closed loop system: a single sample path }
                \label{fig:states_fb}
        \end{subfigure}
             \caption{Plots of example 2 }\label{fig:example2}
\end{figure}
\end{example} 

\section{Conclusions}\label{section:conclusions}
 In this paper we have treated the stochastic
stability and stabilization of a state-dependent jump linear system.
We utilized the stopping times as a pointer to capture the evolution
of the state variable, and used the Dynkin's formula to obtain sufficient
condition in terms of linear matrix inequalities. Using the sufficient
condition, we synthesize a state-feedback controller which stochastically
stabilizes the system. 

\section*{Appendix}

\textbf{Proof of Lemma  \ref{lemma:joint_Markov_process}:} 
We follow the approach used in \cite{hybridge}
and \cite{maobook} to prove. We prove in the sequel that $(x(t),\theta(t))$
is jointly Markovian which is equivalent to state that $(x(t),r(\sigma_t,t),\sigma_t)$
is a joint Markov process, because the evolution of $\theta(t)$ is
captured in $r(\sigma_t,t)$ and $\sigma_t$ from (\ref{eq:sigma_intro})
and (\ref{eq:rt_intro}). Let $\mathcal{F}_{t}$ denote the natural filtration of
$(x(t),\theta(t))$ on the interval $[0,t]$, and $\mathcal{F}^{t}$
denote the natural filtration of $(x(t),\theta(t))$ on the interval
$[t,\infty)$. Let $x(t)$ and $\theta(t)$ from (\ref{sys:system_old})
and (\ref{eq:rt_intro_old}) be denoted , starting from $t_{0}$,
as 
\begin{align*}
x(t)\triangleq x_{t}^{t_{0},x(t_{0}),\theta(t_{0})},\qquad\theta(t)\triangleq\theta_{t}^{t_{0},x(t_{0}),\theta(t_{0})}.
\end{align*}
Denote the joint process $\xi_{t}^{t_{0},\eta(t_{0})}$ starting from
$t_{0}$ as 
\begin{align*}
\xi_{t}^{t_{0},\eta(t_{0})}\triangleq\left(x_{t}^{t_{0},x(t_{0}),\theta(t_{0})},\theta_{t}^{t_{0},x(t_{0}),\theta(t_{0})}\right),
\end{align*}
with $\eta(t_{0})=\left(x_{t_{0}}^{t_{0},x(t_{0}),\theta(t_{0})},\theta_{t_{0}}^{t_{0},x(t_{0}),\theta(t_{0})}\right)$,
which is simply $(x(t_{0}),\theta(t_{0}))$. For $0<s<t$, $\xi_{t}^{s,\eta(s)}$
describes the process on $[s,\infty)$ with $\eta(s)=\left(x_{s}^{s,x(s),\theta(s)},\theta_{s}^{s,x(s),\theta(s)}\right)$,
thus it is $\mathcal{F}^{s}$-measurable. Let $\eta(s)$ be an arbitrary
$\mathcal{F}_{s}$-measurable random variable. For $0<\tau<s<t$,
$\xi_{t}^{\tau,\eta(\tau)}$ can be described as a $\mathcal{F}^{s}$-measurable
process on $[s,\infty)$ with initial condition $\xi_{s}^{\tau,\eta(\tau)}$.
Thus we can write 
\begin{align}
\xi_{t}^{\tau,\eta(\tau)}=\xi_{t}^{s,\xi_{s}^{\tau,\eta(\tau)}},\qquad\text{for}\quad0<\tau<s<t.
\end{align}
Let $B$ be a set in the $\sigma$-algebra of Borel sets on $\mathbb{R}^{n}\times S$.
Then 
\begin{align*}
Pr\{\xi_{t}^{\tau,\eta(\tau)}\in B|\mathcal{F}_{s}\}=E\left[I_{B}\left(\xi_{t}^{\tau,\eta(\tau)}\right)|\mathcal{F}_{s}\right] & =E\left[I_{B}\left(\xi_{t}^{s,\xi_{s}^{\tau,\eta(\tau)}}\right)|\mathcal{F}_{s}\right]\\
 & =E\left[I_{B}\left(\xi_{t}^{s,\eta(s)}\right)\right]_{\eta(s)=\xi_{s}^{\tau,\eta(\tau)}}\\
 & =Pr\{\xi_{t}^{s,\eta(s)}\in B\},
\end{align*}
which completes the proof. \hfill{}$\square$

\[
\vphantom{}
\]

\noindent \textbf{Proof of Lemma \ref{lemma:stopping_time}:} Let $\mathcal{G}_{t}$ denote
the natural filtration of $(x(t),r(\sigma_t,t))$ on the interval
$[0,t]$. Consider 
\begin{align*}
\{\tau_{0}\le t\}= & \bigcup_{s\in\mathbb{Q}\cap[0,t]}\{\Phi_{i_{0}}(s,0)x(0)\notin C_{i_{0}}\}\\
= & \,\,\,\,\Omega\setminus\bigcap_{s\in\mathbb{Q}\cap[0,t]}\{\Phi_{i_{0}}(s,0)x(0)\in C_{i_{0}}\}.
\end{align*}
 From the above argument, observe that, each event $\{\Phi_{i_{0}}(s,0)x(0)\in C_{i_{0}}\}$
is $\mathcal{G}_{t}$-measurable for all $s\in\mathbb{Q}\cap[0,t]$.
Consequently the event $\{\tau_{0}\le t\}$ is also $\mathcal{G}_{t}$-measurable,
as the complement of the intersection of $\mathcal{G}_{t}$-measurable
events are also $\mathcal{G}_{t}$-measurable. Thus $\tau_{0}$ is
a stopping time. The similar arguments are applied to $\tau_{1},\,\tau_{2,}\cdots.$\hfill{}
$\square$
\bibliographystyle{plain}
\bibliography{References_syscon_13}

\begin{thebibliography}{10}

\bibitem{bergman}
B.~Bergman.
\newblock Optimal replacement under a general failure model.
\newblock {\em Advances in Applied Probability}, 10 (2):431 -- 451, 1978.

\bibitem{boukasbook}
E.K. Boukas.
\newblock {\em Stochastic switching systems: Analysis and design}.
\newblock Birkh{\"a}user Boston, 2005.

\bibitem{boukas1990}
E.K. Boukas and A.~Haurie.
\newblock Manufacturing flow control and preventive maintenance: a stochastic
  control approach.
\newblock {\em IEEE Transactions on Automatic Control}, 35 (9):1024 -- 1031,
  1990.

\bibitem{costabook2012}
O.L.V. Costa, M.D. Fragoso, and M.G. Todorov.
\newblock {\em Continuous-time Markov Jump Linear Systems}.
\newblock Springer, 2012.

\bibitem{fangphd}
Y.~Fang.
\newblock Stability analysis of linear control systems with uncertain
  parameters.
\newblock {\em Ph.D. dissertation, Dept. Syst., Control, Ind. Eng., Case
  Western Reserve Univ., Cleveland, OH}, Jan. 1994.

\bibitem{feng}
X.~Feng, K.~A. Loparo, Y.~Ji, and H.~J. Chizeck.
\newblock Stochastic stability properties of jump linear systems.
\newblock {\em IEEE Transactions on Automatic Control}, 37:38 -- 53, 1992.

\bibitem{filar2001}
J.~A. Filar and A.~Haurie.
\newblock A two-factor stochastic productin model with two time scales.
\newblock {\em Automatica}, 37:1505 -- 1513, 2001.

\bibitem{motivation1}
M.~Giorgio, M.~Guida, and G.~Pulcini.
\newblock A state-dependent wear model with an application to marine engine
  cylinder liners.
\newblock {\em Technometrics}, 52 (2):172 -- 187, 2010.

\bibitem{yji}
Y.~Ji, H.~J. Chizeck, X.~Feng, and K.~A. Loparo.
\newblock Stability and control of discrete-time jump linear systems.
\newblock {\em Control Theory and Advanced Technology}, 7:247 -- 270, 1991.

\bibitem{motivation2}
C.~J. Kim, J.~Piger, and R.~Startz.
\newblock Estimation of {M}arkov regime-switching regression models with
  endogenous switching.
\newblock {\em Journal of Econometrics}, 143:263 -- 273, 2008.

\bibitem{kozin}
F.~Kozin.
\newblock A survey of stability of stochastic systems.
\newblock {\em Automatica}, 5:95 -- 112, Jan. 1969.

\bibitem{hybridge}
J.~Krystul and H.~A.~P. Blom.
\newblock Generalized stochastic hybrid processes as strong solutions of
  stochastic differential equations.
\newblock {\em Hybridge report D 2, 2005}.

\bibitem{maobook}
X.~Mao and C.~Yuan.
\newblock {\em Stochastic differential equations with {M}arkovian switching}.
\newblock Imperial College Press, 2006.

\bibitem{mariton}
M.~Mariton.
\newblock {\em Jump Linear Systems in Automatic Control}.
\newblock Marcel Dekker, 1990.

\bibitem{sri}
R.~Srichander and B.K. Walker.
\newblock Stochastic stability analysis for continuous fault tolerant control
  systems.
\newblock {\em International Journal of Control}, 57(2):433 -- 452, 1993.

\bibitem{sworder1973}
D.D. Sworder and V.G. Robinson.
\newblock Feedback regulators for jump parameter systems with state and control
  dependent transition rates.
\newblock {\em IEEE Transactions on Automatic Control}, 18:355 -- 359, 1973.

\bibitem{karimi}
Z.~Wu, M.~Cui, P.~Shi, and H.~R. Karimi.
\newblock Stability of stochastic nonlinear systems with state-dependent
  switching.
\newblock {\em IEEE Transactions on Automatic Control}, 58 (8):1904 -- 1918,
  2013.

\bibitem{yin2009hybrid}
G.~G. Yin and C.~Zhu.
\newblock {\em Hybrid switching diffusions: properties and applications},
  volume~63.
\newblock Springer, 2009.

\end{thebibliography}
 
\end{document}